\def\SKETCH{}
\def\EX{}

\documentclass{article}
\usepackage[utf8]{inputenc}

\usepackage{amsthm}
\usepackage[T1]{fontenc}    
\usepackage{hyperref}       
\hypersetup{colorlinks=true,linkcolor=blue,citecolor=blue}
\usepackage{url}            
\usepackage{booktabs}       
\usepackage{amsfonts}       
\usepackage{nicefrac}       
\usepackage{microtype}      
\usepackage{stmaryrd}
\usepackage{xspace}
\usepackage[ruled]{algorithm2e} 

\usepackage{caption}
\usepackage{enumitem}

\usepackage{geometry}
\usepackage[authoryear]{natbib}

\usepackage{comment}

\mathchardef\shyp="2D
\newcommand{\kMAJ}[1]{#1\shyp\mathrm{SORT}}
\newcommand{\kREP}[1]{#1\shyp\REP}

\def\newpar{\vspace{-2mm}\paragraph}

\newtheorem{theorem}{Theorem}
\newtheorem{definition}{Definition}
\newtheorem{lemma}[theorem]{Lemma}
\newtheorem{proposition}[theorem]{Proposition}
\theoremstyle{remark}
\newtheorem{remark}{Remark}
\newtheorem{conjecture}{Conjecture}
\newtheorem{corollary}{Corollary}
\newtheorem{example}{Example}

\usepackage{amsmath}
\usepackage{amssymb}
\usepackage{bm}
\usepackage{amsthm}
\usepackage{mathtools}
\usepackage{blindtext}
\usepackage{graphicx}
\graphicspath{ {images/} }
\usepackage{verbatim}
\usepackage{xcolor}
\usepackage{tikz}
\usepackage{gensymb}
\usepackage{textcomp}
\usepackage{verbatim}
\usepackage{array}

\def\calX{\mathcal{X}}
\newcommand{\tup}[1]{\left\langle#1\right\rangle}
\newcommand{\ind}[1]{\left\llbracket #1 \right\rrbracket}
\DeclareMathOperator{\argmin}{argmin}

\newcommand{\N}{\mathbb N}

\newcommand{\E}{\mathbb{E}}
\newcommand{\V}{\mathbb{V}}
\renewcommand{\P}{\mathbb{P}}
\newcommand{\AR}{\mathrm{AR}}
\newcommand{\SC}{\mathrm{SC}}
\newcommand{\opt}{\mathrm{opt}}
\newcommand{\RD}{\mathrm{RD}}
\newcommand{\MAJ}{\mathrm{MAJ}}
\newcommand{\REP}{\mathrm{wSORT}}
\newcommand{\REG}{\mathrm{Regret}}
\newcommand{\iid}{\mathrm{i.i.d.}}
\newcommand{\sym}{\mathrm{anym}}

\title{Representative Committees of Peers\footnote{The authors are grateful to Nisarg Shah for valuable suggestions.\newline 
The work of Reshef Meir, Fedor Sandomirskiy, and Moshe Tennenholtz is supported by the European Research Council (ERC) under the European Union's Horizon 2020 research and innovation program (\#740435). \newline 
Fedor Sandomirskiy is partially supported by
Grant 19-01-00762 of the Russian Foundation for Basic Research, and by the Basic Research Program of the National Research University Higher School of Economics.} }
\author{Reshef Meir, Fedor Sandomirskiy, Moshe Tennenholtz}
\author{ \large Reshef Meir\thanks{Technion, Haifa (Israel).}\ \ \ %
Fedor Sandomirskiy\thanks{Technion, Haifa (Israel) and Higher School of Economics, St.Petersburg (Russia).} \ \ \ %
Moshe Tennenholtz\thanks{Technion, Haifa (Israel). }}

\date{}

\newcommand{\kibitz}[2]{\ifdefined\DRAFT{\color{#1}{#2}}\fi}
\newcommand{\rmr}[1]{\kibitz{red}{[RESHEF:#1]}}
\newcommand{\fsy}[1]{\kibitz{blue}{[FEDOR:#1]}}

\definecolor{ForestGreen}{rgb}{.13,.54,.13}
\definecolor{BrickRed}{rgb}{.80,.26,.33}

\ifdefined\DRAFT
\newcommand{\new}[1]{{\color{ForestGreen}{{#1}}}}

\else
\newcommand{\new}[1]{#1}

\fi

\begin{document}

\maketitle

\fsy{Do you like the title? Maybe ``Selecting representative committees of peers to cope with misty future''}\rmr{not so much, but maybe add ``sortition" somewhere?}\fsy{Representative Committees of Peers via Sortition?}

\begin{abstract}
A population of voters must elect representatives among themselves to decide on a sequence of possibly unforeseen binary issues. Voters care only about the final decision, not the elected representatives. The disutility of a voter is proportional to the fraction of issues, where his preferences disagree with the decision.

While an issue-by-issue vote by all voters would maximize the social welfare, we are interested in how well the preferences of the population can be approximated by a small committee.

We show that a $k$-sortition (a random committee of $k$ voters with the majority vote within the committee) leads to an outcome within the factor $1+O\left(\frac{1}{\sqrt{k}}\right)$ of the optimal social cost
for any number of voters $n$, any number of issues $m$, and any preference profile.

For a small number of issues $m$, the social cost can be made even closer to optimal by delegation procedures that weigh committee members according to their number of followers. However, for large $m$, we demonstrate that the $k$-sortition is the worst-case optimal rule within a broad family of committee-based rules that take into account metric information about the preference profile of the whole population.
\end{abstract}

\fsy{This are dummy Intro \& Abstract. Please, rewrite/improve! Some important references are missing.}

\section{Introduction}\label{sec_intro}
How well do committees represent preferences of the underlying population? How to select committees optimally? 
This paper aims to answer both questions  taking the design perspective on committee-selection procedure and the choice of inter-committee voting rule.

The mainstream social-choice literature
considers preferences on the set of  candidates to be the primitive of the model. In contrast, we assume that voters have preferences on \emph{ decisions} of the committee. 
This richer structure provides a
natural way to assess how well the committee represents electorate's preferences.\rmr{I think this whole temporal aspect is confusing.  I suggest something like: ``We consider a model where voters ultimately care about the decisions taken, and where electing representatives in only means to get to this goal"}\fsy{See my comment below on this page.}

We consider a population that is going to face a sequence of binary issues to decide on  (e.g., academics select a committee to decide whether to run a new journal, where to organize the next workshop, whether to conduct it online because of the pandemic of COVID-19). Once a new issue emerges, each voter can possibly form his preference about the best alternative. This allows us to consider a ``socially optimal" alternative that would be selected at the whole-population referendum, had such a referendum been conducted.\footnote{We assume that an individual's preferences of issues are independent, or \emph{separable}.} However, engaging in a frequent referenda may be a heavy burden for the population. 

This is part of the reason why historically, most of the world has moved from direct democracy to some form of representative democracy (by committees or parliaments), whereas more recently there is much interest in increased public participation via technological means \citep{dalton2001public,brill2019interactive,allen2019cryptodemocracy}. \fsy{BETTER REFERENCES?}\rmr{anything but blockchain?}

We follow the representative democracy approach in this paper, assuming there is a strict limit on the number of representatives that can actively vote on all issues. We then ask if and how a representative committee can be selected from the general population.

We evaluate how well the committee aggregates preferences of a given population by its social cost, where the cost of a voter is proportional to the number of issues on which he disagrees with the final decision. 
Following the standard worst-case approach in approximate mechanism design, 
the approximation ratio is defined as the worst-case ratio between the (expected) social cost of the elected committee and the optimal social cost. The latter is obtained by an issue-by-issue majority vote of the entire population.
We impose no restriction on preference profiles of the population.

The simplest way to select a committee is to select $k$ individuals  uniformly at random from the population. This way is known as ``sortition" and it has been discussed---although not frequently used---since ancient times~\cite{dowlen2017political}.

Note that we did not say explicitly \emph{how} the selected committee is supposed to make a decision. Indeed, a natural approach is to use a majority vote inside the committee but we may also consider other alternatives. 

Our main goal in this paper is to understand the best approximation ratio that can be obtained by a committee, and what selection rules and/or internal voting rules are optimal.


\rmr{I think that the next paragraph is mainly confusing} \fsy{I am sure this paragraph and the whole temporal interpretation would be critical if we submitted to the econ journal. I have much less intuition about CS/AI venues. So if you think we don't get credit for it, feel free to delete.}\rmr{for a journal we can have a broader discussion on the temporal aspects. For now I think simplicity is the key.}

\subsection{Our results} 
Our main result is a characterization of the approximation ratio of $k$-sortition:
For a committee of size $k=3$ the ratio is equal to $1.316$ (Example~\ref{ex_3_MAJ}) and behaves as $1+\Theta\left(\frac{1}{\sqrt{k}}\right)$, when $k$ increases (Theorem~\ref{th_approx_ratio_arbitrary_k}). As a corollary of this result, we infer that the optimal committee size $k$ for a population of size $n$ is of the order of $n^\frac{2}{3}$ given fixed preference-elicitation costs (Corollary~\ref{cor_opt_size}).

Can we do better than $k$-sortition?
It turns out that for a small number of issues $m$, the approximation ratio can be made exponentially close to $1$ in $k$ by weighing each committee member according to the fraction of the electorate it best represents (Proposition~\ref{prop_delegation_small_m}). However, this improvement  disappears in a realistic scenario of large number of issues $m$, moreover, the suggested average-proximity weighing becomes harmful (Proposition~\ref{prop_delegation_is_bad}).

We complete the results above by showing that without further restrictions,
$k$-sortition is the best rule among a broad family of committee-selection and inter-committee voting rules (Theorem~\ref{theorem_general}).

To summarize, in unstructured environments, $k$-sortition is a compelling choice for representing preferences of the electorate.
 \fsy{``Of course, it has certain limitations, which we discuss at the end of the paper.'' Either delete this phrase or discuss in Concluding section}
\paragraph{Structure of the paper}
The model is introduced in Section~\ref{sec_model}. We analyze the benchmark rule, $k$-sortition in Section~\ref{sect_uniform_majority}. Section~\ref{sec_delegation} is devoted to improving the outcome of $k$-sortition for small number of issues by using additional information about the proximity of committee members and voters. In Section~\ref{sec_general}, we show that $k$-sortition is worst-case optimal within a broad family of rules. Limitations of the model and remaining questions are discussed in Section~\ref{sec_conclusions}. Most of the proofs are only sketched in the main body of the paper; all the details can be found in Appendices~\ref{app_proofs_unif_maj},~\ref{app_proofs_delegation}, and~\ref{app_proofs_general}.

\subsection{Related literature}
Finding good voting rules for committee-selection and even formalizing what is ``good'' is still a big challenge for the social-choice literature, see a recent survey~\citep{faliszewski2017multiwinner}; this literature takes a normative approach and aims to characterize voting rules by the combination of desired properties known as axioms.  Preferences over candidates are considered to be the primitive of the model despite that, as argued by political scientists, preferences on candidates originate from the electorate's preferences on future decisions that a candidate would take if elected~\citep{austen1988elections}. This limitation does not allow to capture the objective of selecting committees that optimally represent electorate's preferences on future policies.

The idea of selecting representatives as a small random sample from the electorate has its roots in Athenian democracy (see e.g.~\citep{hansen1999athenian}).  
Recently this idea gained popularity \citep{dowlen2017political,cheng2017people, cheng2018distortion, guerrero2014against} for its fairness, representativness, and high barriers to various manipulations including vote-buying \citep{parkes2017thwarting, jamroga2019risk, gersbach2017sophisticated}. It was also deployed by at least one startup \citep[\href{https://rsvoting.org/}{https://rsvoting.org/}]{chaum2016random}. Preliminary experimental data suggests that random-sample voting is suitable for political elections, see \citep{fishkin2018deliberative} and \citep[Chapter 7]{blanchard2019usability}.  Trustworthy and secure large-scale implementation of these ideas leads to cryptographic challenges \citep{lenstra2015random, basin2018alethea}.
We emphasize that in this work fairness (in the sense of representation for minorities) is not one of our goals, as the selected committee is only an intermediate step for a  final decision that applies to the entire population.


One stream of the literature on multiple referenda deals with the dependency among issues in voters' preferences (see e.g.~\cite{lang2016voting}); or with restrictions on the valid outcomes, as in \emph{judgement aggregation}~\cite{pauly2006logical}. We follow a different stream of the literature, making the simplifying assumption of that the full preferences of each voter are fully captured by her position in some metric space (in our case, the binary cube with the Hamming distance)~\cite{border1983straightforward,procaccia2009approximate,meir2012algorithms,goel2017metric, anshelevich2018approximating}.  

Measuring the performance of a voting rule by its approximation ratio (or distortion, which is an analogue of the approximation ratio for a given set of candidates) was suggested in~\citep{procaccia2006distortion}. It was later applied by \citet{procaccia2009approximate} for 
 facility-location in metric spaces and became a gold standard in economic-design literature.
In this line of papers, the optimal outcome is compared to the  best possible outcome subject to some constraint on the voting rule (e.g. that it is strategyproof, or uses only ordinal information). Indeed, in the special case of a singleton ($k=1$),  $k$-sortition boils down to the familiar \emph{random dictator} voting rule, whose approximation ratio for large populations is $2$. 


\citet{cheng2018distortion}  also consider voting on a metric space, where  candidates are uniformly sampled from the set of voters, and characterize the class of scoring rules having bounded distortion.  Our paper is different in many aspects: it analyzes the decisions of the elected committee, allows for fairly general committee-selection procedure, e.g., the sampling procedure may depend on preferences of the electorate, and achieves the approximation-ratio close to $1$.\rmr{I still need to better understand this paper}

\fsy{Shall we use the term "distortion" instead of the approximation ratio? In our model they are the same}

\medskip
The fact that the approximation ratio of $k$-sortition goes to $1$ for large committee size, can be regarded as an extension of the famous Condorcet Jury theorem, which claims that if each voter receives a noisy signal about the ``right'' alternative, the majority vote of a large population reveals the ground truth. In this interpretation, preferences of a random committee member are seen as a noisy estimate of preference of the majority. An important difference is that the jury theorem requires the population to be far enough from a tie~\citep{paroush1998stay}, while our bounds are applicable to all preference profiles. \rmr{I wonder if we can make a more explicit comparison with CJT: Suppose we have $n$ voters each with a correct signal w.p. $p\in[0,1]$. We have the same requirement of the committee whether the signal is weak or strong.

However with a 1-issue voting, expected  social cost means that with an almost even partition of the population (=``weak signal") we also require less of the voting outcome. }



\newpar{Delegation}
A compromise between the direct and representative democracies is provided by proxy voting \citep{alger2006voting, green2015direct, cohensius2017proxy} and liquid democracy \citep{kahng2018liquid, goelz2018fluid}: each voter has an opportunity to engage in the decision-making directly but, since he may be non-motivated enough or unavailable, there is an option to specify a representative thus delegating the vote to the proxy. The idea of weighing the committee members, used in Section~\ref{sec_delegation}, is inspired by this line of research. 
Below we provide a detailed overview of several papers on proxy voting with  multi-issue setting sharing some similarity with ours.

\citet{pivato2020weighted}  consider a fixed committee with delegation to the closest representative and evaluate its performance on a sequence of i.i.d. issues. They demonstrate that in the limit of large electorate, the decision of the committee always matches the alternative preferred by the majority. This conclusion holds even for committees of size $1$ because authors impose additional strong assumptions on the alignment of electorate's and committee members' preferences.
Our approach significantly differs: we are interested in the worst-case guarantees, work with a fixed population of voters, allow them to have arbitrary preferences, and the only randomness in our model comes from the randomization performed by  voting rules.

The setting in \citep{skowron2015we} is similar to \citep{pivato2020weighted}: implicitly assumed i.i.d. issues and strong restrictions on the alignment of preferences of the committee and the electorate. The paper discusses how the inter-committee voting rule affects the optimal selection of the committee. We are interested in the optimal selection of the pair (a committee and an inter-committee voting rule), make no assumptions on the preference profile, and rely on the worst-case analysis giving the robust guarantees.

In contrast to the two aforementioned papers, \citet{abramowitz2018flexible} impose no restrictions on preferences and propose an interesting continuous way to mix direct and proxy voting by allowing voters to readjust the weights of committee members for each issue. 
In Section~\ref{sec_general}, we consider a family of rules containing the proposal of \cite{abramowitz2018flexible} and demonstrate the worst-case optimality of $k$-sortition within this family.




\section{The model}\label{sec_model}
There is a population $[n]=\{1,2,\ldots,n\}$ of $n$ voters who are going to face a sequence \rmr{``sequence" implies some temporal order. Is that the word we want?}
\fsy{Yes, I used the word sequence to emphasize the sequential nature of decisions. This interpretation helps to justify issue-wise voting within the committee and the informational parsimony w.r.t. preferences of the whole population. Do you find this confusing?}\rmr{yes, I think it is ok to use in the the motivational parts, but not in the technical part. ultimately every voter is just a binary vector.}\fsy{A binary vector $=$ a sequence of zeros and ones... So to me it does not sound confusing, but feel free to correct it. Indeed, I don't not know the specifics of our audience.}
of $m$ binary issues $[m]=\{1,2,\ldots,m\}$ to decide on. 
Preferences of each voter $i\in [n]$ are represented by a vector $x_i\in \{0,1\}^m$, where $x_{i,j}=1$ if $i$ prefers the alternative $1$ for issue $j$ and $x_{i,j}=0$ if the alternative $0$ is preferred. The preference  profile of the population is thus given by an $n\times m$--matrix $X=(x_{i,j})_{i\in [n],j\in[m]}$ of zeros and ones.

For a pair of vectors $z,z'\in\{0,1\}^m$, we define the distance $d(z,z')$ between them to be the number of issues where they disagree, i.e., the Hamming distance: $d(z,z')=\sum_{j\in [m]}|z_{j}-z_j'|.$ For a pair of voters $i,k\in[n]$, the distance between their preferences $d(x_i,x_j)$ captures the overall alignment of their tastes. 

\begin{definition}[voting rules]
A voting rule $f$ specifies the outcome vector $z\in \{0,1\}^m$  for each preference profile profile $X$.

We allow for randomization and so a voting rule maps $X$ to a probability distribution $f(X)\in \Delta\big(\{0,1\}^m\big)$ according to which $z\in \{0,1\}^m$ is then chosen. 
\end{definition} 
Here and below $\Delta(A)$ denotes the set of all probability measures on $A$.

\paragraph{Classes of committee voting rules.} We are interested in voting rules that can be represented as a two-stage procedure: first voters select a certain committee of peers $C\subset [n]$ and then the committee members vote to determine the outcome for each of the issues. The restriction of a preference profile $X$ to committee members $C\subset [n]$ is denoted by $X\vert_C$;
\new{for definiteness, we assume that the order of rows in $X\vert_C$ is the same as in $X$, i.e., the row corresponding to a committee member with lower index comes first.} 
\begin{definition}[committee voting rules]\label{def:cvr}
A committee voting rule is given by a function $g$ that maps preference profile $X$ to a probability distribution over pairs $(C,h)$, where a committee $C$ is a subset of $[n]$ and $h:\{0,1\}^{|C|\times m}\rightarrow \{0,1\}^m$ is a voting rule. The outcome of $g$ is computed as follows: first the pair $(C,h)$ is chosen with probability $g(X)$ and then the outcome-vector $z$  is obtained by applying $h(X\vert_C)$.

We say that $g$ is a $k$-committee rule if $|C|=k$ with probability $1$. 
\end{definition}
\new{Any committee rule is a voting rule, and any voting rule can be seen as a $k$-committee rule with $k=n$, i.e., the whole population plays a role of the committee. We are interested in the opposite scenario, when the committee size $k$ is small compared to the total number of voters, the total number of voters $n$.} 


Without further restrictions, 
every voting rule $r$ (even randomized) can be represented as a $1$-committee rule, where $C$ contains a single arbitrary voter and $h$ is the fixed rule $h(X|_C)\equiv r(X)$.

We therefore consider two natural restrictions on committee rules: the selection phase can only be based on pairwise distances; and in the voting phase, the committee decides on each issue separately (formal definitions below).

 Indeed, the idea behind committee rules is that they avoid elicitation of detailed preferences of large population and rely on a small number of committee-members instead;\footnote{the detailed information about preferences may be unavailable at the time of the committee selection if $[m]$ contains some unforeseen issues.}
  and the decision making process of the committee itself should be simple and transparent to maintain a clear connection between representatives and the voters they represent, and to enable voting on arbitrary issues as they come along.

\new{For a preference profile $X$, we denote by $D(X)$ the matrix of pairwise distances $D(X)=\big(d(x_i,x_j)\big)_{i,j\in[n]}$. We focus on those committee rules that base the choice of the committee and of the inter-committee rule on information encoded in $D(X)$ only.
\begin{definition}[distance-based committee rules]
A committee voting rule $g$ is \emph{distance-based} if the probability distribution $g(X)$ over pairs $(C,f)$ of a committee and an inter-committee voting rule depends only on matrix of pairwise distances $D(X)$. In other words, $g(X)=g(X')$ whenever $D(X)=D(X')$.
\end{definition}
\begin{example}[a distance-based $k$-committee rule.]\label{ex_distance-based}
Consider the following $k$-committee rule. It picks a committee $C$ that minimizes the total distance to all the voters $\sum_{i\in[n]}\min_{i'\in C} d(x_i,x_{i'})$ (if there are several such committees, it randomizes over them uniformly) and the inter-committee rule $h$ is the weighted majority vote on every issue, where the weight of a member $c$ is given by $w_c=\sum_{i\in [n]} d(x_i,x_c)$. So for issue $j$, the outcome $z_j=1$ whenever $\sum_{c\in C:\, x_{c,j}=1} w_c>\frac{1}{2}\sum_{c\in C} w_c$; in case of opposite strict inequality, $z_j=0$ and $z_j\in\{0,1\}$ is picked at random in case of a tie.
\end{example}
\begin{definition}[issue-wise voting rule]
A voting rule $f$ is issue-wise  if it treats all  issues separately i.e., for any pair of preference profiles $X,X'$ that coincide on issue $j$ and the corresponding outcome vectors $z\sim f(X)$ and $z'\sim f(X')$, we have $\P(z_j=1)=\P(z_j'=1)$. 
\end{definition}
The latter property is also known as \emph{independence of irrelevant alternatives (IIA)} (see e.g.~\citep{pauly2006logical}). \fsy{For me, the connection to IIA is not so clear. Different issues may be related one to another. For example, if the third issue is XOR of the first two, then by looking at preferences of voters on the first two issues the rule may improve the outcome on the third...}
\rmr{This is exactly IIA, see e.g. here \url{https://link.springer.com/content/pdf/10.1007/s10992-005-9011-x.pdf} :
``IIA demands that the value that an aggregation function attaches to a
proposition is dependent only on the individual values concerning that
proposition."
}
\begin{definition} \rmr{maybe call it ``committee-IIA'' (CIIA) or ``committee-issue-wise'' (CIW)?}
A committee rule $g$ is has issue-wise inter-committee vote if for any preference profile $X$ the probability distribution $g(X)$ is supported on pairs $(C,h)$ such that the inter-committee voting rule $h$ is issue-wise.
\end{definition}
All the committee rules we consider will be distance-based with issue-wise inter-committee vote. 
Note issue-wise inter-committee vote does not imply that the rule itself is issue-wise since the distribution of the pair $(C,h)$ depends on preferences of the whole population over all issues through $D(X)$. E.g., the rule from  Example~\ref{ex_distance-based} is not issue-wise but has issue-wise inter-committee vote. 
An example of a committee voting rule that violates the definition, is when $h$ approves the three issues with the largest support in $X|_C$ and rejects all others. 
}
\rmr{I added above an example showing why without the distance-based assumption, any voting rule is a 1-committee rule. I tried to think of a similar example as to why committee-IIA may be necessary to avoid trivialities. nothing so far... (maybe you are right and it is not necessary)}
\paragraph{Inefficiency of voting rules.}  Informational parsimony of distance-based committee rules may lead to selecting suboptimal outcomes. To quantify this inefficiency, we employ the utilitarian approach. We define the disutility of a voter $i$ for an outcome-vector $z\in \{0,1\}^m$ as the total number of issues, where $i$'s preferences and $z$ disagree, i.e., this disutility equals to the distance $d(x_i,z)$.

The utilitarian social cost of an outcome vector $z$ is given by the sum of disutilities, and the social cost of a voting rule $f$ on a preference profile $X$ is defined as the expected social cost of its outcome:
$$\SC(z)=\sum_{i\in [n]} d(x_i,z), \ \ \ \ \ \SC\big(f(X)\big)=\E_{z\sim f(X)} \left[SC(z)\right].$$ 
Denote by $z^{\opt}=z^{\opt}(X)$ the socially-optimal outcome, i.e., the one with the minimal cost
$$z^{\opt}=\argmin_{z\in \{0,1\}^m} \SC(z).$$
\begin{definition}
The approximation ratio of a voting rule $f$ is the worst-case ratio of  the expected social cost of its outcome $z$ and the optimal social cost
\begin{equation}\label{eq_def_AR}
\AR_{n,m}(f)=\max_{X\in \{0,1\}^{n\times m}}\left(\frac{\SC\big(f(X)\big)}{\SC\big(z^{\opt}\big)}\right) \in [1,+\infty]    
\end{equation}
The maximum is over all preference profiles $X$ with fixed numbers  $n,m$ of voters and issues. If the denominator is zero (this happens for unanimous preference profiles), the following agreement is used: $\frac{0}{0}=1$ and $\frac{C}{0}=+\infty$ for $C>0$. 

When we drop some of the subscripts $n$ and/or $m$, this means considering their worst-case values, i.e.,  $\AR_{n}=\sup_{m\in \N} \AR_{n,m}$, $\AR_{m}=\sup_{n\in \N} \AR_{n,m}$, and $\AR=\sup_{n,m\in \N} \AR_{n,m}$,
\end{definition}
\ifdefined\EX
\begin{example}[The simple majority rule $\MAJ$]\label{ex_MAJ}
The outcome $z_j$ for each issue $j$ is determined on the whole-population referendum: $z_j=0$ if $\sum_{i\in [n]} x_{i,j}< \frac{n}{2}$;  for the opposite strict inequality, $z_j$ equals $1$; and $z_j$ is $0$ or $1$ equally likely in case of a tie.
The resulting majority rule denoted by $\MAJ$ has the approximation ratio $\AR_{n,m}(\MAJ)=1$, i.e., it always select the socially-optimal outcome. Indeed, minimization of the social cost of $\SC(z)$ splits into a family of issue-wise minimization problems $\min_{z_j}\sum_{i\in [n]} |x_{i,j}-z_j|$ for each $j\in [m]$ and the minimum is achieved if $z_j$ matches the majority of $(x_{i,j})_{i\in [n]}$.
Optimality of $\MAJ$ comes at the cost of huge burden imposed on the population: to compute the outcome we may need the preferences of the whole population on all the issues.
\end{example}
\begin{example}[The random dictatorship $\RD$]\label{ex_RD}
The random dictator rule $\RD$ is a benchmark example in the voting theory. In our setting it works as follows: a voter $i\in [n]$ is selected uniformly at random and $i$'s preference dictates the outcome of each referendum, i.e., $z=x_i$. Note that $\RD$ is an example of a $1$-committee rule.

It is known that $\RD$ has the approximation ratio $\AR_n(\RD)\leq 2-\frac{2}{n}$ for voting in general metric spaces~\citep{meir2012algorithms,anshelevich2017randomized}.
Our setting can be regarded as a particular metric space $\{0,1\}^m$ with the Hamming distance $d$. However this does not improve the approximation ratio.

Let us compute the approximation ratio for the one-issue case. Without loss of generality, we can restrict our attention to profiles $X$ such that the majority prefers the alternative $0$ but not unanimously (for unanimity $\SC(z^{\opt})=\SC(\RD(X))=1$). Denote  the number of supporters of the alternative $1$ by $n_1=\sum_{i\in [n]} x_{i,1}$. By the assumption, $1\leq n_1\leq \frac{n}{2}$. The social costs are $\SC(0)=n_1$ and $SC(1)=n-n_1$. With probability $\frac{n_1}{n}$ the random dictator supports $1$ and with the complement probability he supports $0$. Therefore, 
$$\frac{\SC(\RD(X))}{\SC(z^\opt)}=\frac{\frac{n_1}{n}\cdot(n-n_1)+ \left(1-\frac{n_1}{n}\right)\cdot n_1}{n_1}=2-\frac{2 n_1}{n}.$$
The maximal value $\AR_{n,1}(\RD)=2-\frac{2}{n}$  is achieved at $n_1=1$.

A simple argument (Lemma~\ref{lm_one_issue_gives_an_upper_bound} below) implies that $\AR_{n,m}(\RD)$ does not depend on $m$ and thus $\AR_{n,m}(\RD)=2-\frac{2}{n}$ for any number of issues.
\end{example}
\fi

\section{Random committees with simple majority rule do good job}\label{sect_uniform_majority}
\rmr{Maybe just ``k-sortition is good"?}
At the end of the last section, we considered the two examples. At one extreme, we have the simple majority rule $\MAJ$, \new{which achieves a socially-optimal outcome but may require the information on preferences of the whole population of voters.}
The random dictatorship $\RD$ is on the other extreme: it is enough to know preferences of one voter only, but $\RD$ can almost double the social cost in the worst-case.

Our goal in this section is to find the middle ground between these two extremes. We construct simple rules that have the approximation ratio close to $1$ and require to learn preferences of a negligible fraction of the population.
The following family of $k$-committee rules is a natural extension of both $\MAJ$ and $\RD$. \rmr{references to similar committee rules?}\fsy{Have you seen some?}
\begin{definition}[\new{$\kMAJ{k}$: $k$-sortition}]
A committee $C\subset [n]$ of size $|C|=k$ is selected uniformly at random. The preferences of the committee members are aggregated using the simple majority rule $h=\MAJ$ (see Example~\ref{ex_MAJ}).
\rmr{tie breaking?  also perhaps we should call it k-sortition} \fsy{Tie-breaking is uniform, see Example~\ref{ex_MAJ}.}
\end{definition}
Note that for a population of $n$ voters, $\kMAJ{n}$ coincides with the simple majority rule $\MAJ$, while $\kMAJ{1}$ coincides with $\RD$.
\ifdefined\EX
\begin{example}[The approximation ratio for $\kMAJ{3}$ and one issue]\label{ex_3_MAJ}
Let's compute the approximation ratio for a random committee of size $3$ and one issue. Similarly to Example~\ref{ex_RD}, we can assume that the profile $X$ has the following structure: the majority of $n\geq 3$ voters supports the alternative $0$, and there are $1\leq n_1\leq \frac{n}{2}$ supporters of the alternative $1$. 
   We denote by $p=\frac{n_1}{n}\in \left(0,\frac{1}{2}\right]$ the probability that a random voter supports $1$.

Let $A$ be the event that the committee with $3$ members selects the suboptimal alternative $1$; $A$ occurs if and only if two or three members support $1$. The probability of this event equals to
$$\P(A)={\footnotesize \frac{\begin{pmatrix}n_1\\3\end{pmatrix}}{\begin{pmatrix}n\\3\end{pmatrix}}+\frac{\begin{pmatrix}n_1\\2\end{pmatrix}\cdot \begin{pmatrix}n-n_1\\1\end{pmatrix}}{\begin{pmatrix}n\\3\end{pmatrix}}}=
\frac{p\cdot\left(p-\frac{1}{n}\right)\left(3-2p-\frac{2}{n}\right)}{\left(1-\frac{1}{n}\right)\left(1-\frac{2}{n}\right)}.$$
The probability $\P(A)$ tends to $p^2(3-2p)$ for fixed $p$ and large $n$. For $p\leq \frac{1}{2}$ and any $n$, this limiting value provides an upper bound: $\P(A)\leq p^2(3-2p)$.

Taking into account that $\SC(0)=n_1$ and $\SC(1)=n-n_1$, we obtain
$$\frac{\SC\big(\kMAJ{3}\big)(X)}{\SC(z^\opt)}=\frac{(1-\P(A))n_1+\P(A)\cdot(n-n_1)}{n_1}=1+\frac{\P(A)\cdot (1-2p)}{p}.$$
Taking supremum over $p\in \left(0,\frac{1}{2}\right]$ results in an upper bound on the approximation ratio
\begin{equation}\label{eq_3_MAJ_bound}
\AR_{n,1}\big(\kMAJ{3}\big)\leq 1+\sup_{p\in \left(0\frac{1}{2}\right]}\left[p(3-2p)(1-2p)\right]=1+\frac{7\sqrt{7}-10}{27}.
\end{equation}
The maximum is attained at $p^*=\frac{4-\sqrt{7}}{6}\approx 0.226$. It is easy to see that if we define $n_1=\big\lfloor p^*\cdot n\big\rfloor$ and let $n$ go to infinity, the ratio of the social costs converges to the upper bound in~\eqref{eq_3_MAJ_bound}. Thus
$$\AR_{m=1}\big(\kMAJ{3}\big)=1+\frac{7\sqrt{7}-10}{27}\approx 1.316.$$
\end{example}
\fi
The next simple lemma implies that the approximation ratios computed for $\RD$ and $\kMAJ{3}$ in the one-issue case remain the same for any number of issues.
\begin{lemma}\label{lm_one_issue_gives_an_upper_bound}
For any number of voters $n$ and any committee size $k\leq n$, the approximation ratio of $\kMAJ{k}$ does not depend on the number of issues $m$:
\begin{equation}\label{eq_AR_indep_of_m}
\AR_{n,m}\big(\kMAJ{k}\big)=\AR_{n,1}\big(\kMAJ{k}\big)
\end{equation}
\end{lemma}
The lemma is proved in Appendix~\ref{app_proofs_unif_maj}; here we present a sketch.
\ifdefined\SKETCH
\begin{proof}[Proof sketch of Lemma~\ref{lm_one_issue_gives_an_upper_bound}]
The left-hand side of~\eqref{eq_AR_indep_of_m} is bounded by the right-hand side since any worst-case profile with one issue can be cloned $m$ times. To prove the opposite inequality, we observe that the worst-case profile with $m$ issues cannot be worse than its restriction to the worst issue.
\end{proof}
\fi
\begin{corollary}
For a random committee of size $k=3$, the approximation ratio $\AR_{m}\big(\kMAJ{3}\big)\approx 1.316$ for any number of issues $m$.
\end{corollary}
We see that passing from a random dictator to a random $3$-committee significantly improves the approximation ratio. The following theorem covers the case of committees of arbitrary size $k$ and demonstrates that the approximation ratio converges to $1$ fast with the growth of $k$.
\begin{theorem}\label{th_approx_ratio_arbitrary_k}
For any number of voters $n$,  committee size $k\leq n$, and number of issues $m$, the approximation ratio of $\kMAJ{k}$ 
enjoys the following upper bound
\begin{equation}\label{eq_k_MAJ_upper}
    \AR_{n,m}\big(\kMAJ{k}\big)\leq 1+\frac{6\exp\left(-\frac{1}{2}\right)}{\sqrt{k}}\approx 1+\frac{3.639}{\sqrt{k}}.
\end{equation}

This upper bound has the right order of magnitude as a function of $k$. The lower bound is
\begin{equation}\label{eq_k_MAJ_lower}
\AR_{n,m}\big(\kMAJ{k}\big)\geq 1+\frac{2\cdot\left( \Phi\left(-1\right)-\frac{1}{\sqrt{k}}\right)}{\sqrt{k}},
\end{equation}
\new{provided that $2(k+1)^2\leq n$.} Here
 $\Phi(t)=\frac{1}{\sqrt{2\pi}}\int_{-\infty}^{t} \exp\left(-\frac{y^2}{2}\right)\, dy$ is the standard Gaussian cumulative distribution function.
\end{theorem}
The proof is rather long because of technicalities and is postponed to Appendix~\ref{app_proofs_unif_maj}; here we sketch the main ideas.
\ifdefined\SKETCH
\begin{proof}[Proof sketch of Theorem~\ref{th_approx_ratio_arbitrary_k}]
Lemma~\ref{lm_one_issue_gives_an_upper_bound} allows to focus on one-issue case. Similarly to Example~\ref{ex_3_MAJ}, we aim to find the worst-case profile by maximizing over the fraction $p\in \left(0,\frac{1}{2}\right]$ of voters supporting the alternative $1$. Given $n$, $p$, and $k$, the number of supporters of the alternative $1$ in a random committee has the hypergeometric distribution with parameters $(n,p\cdot n, k)$. Despite hyper-geometric random variable is not a sum of i.i.d. contributions, it satisfies the Hoeffding tail-bound, familiar by the i.i.d. case. This bound and optimization over $p$ leads to upper bound~\eqref{eq_k_MAJ_upper}. Note that the worst-case $p=\frac{1}{2}-\Theta\left(\frac{1}{\sqrt{k}}\right)$; however, it takes additional work to exclude small $p$ because the denominator in the approximation ratio~\eqref{eq_def_AR} vanishes as $p\to 0$.

For the lower bound, we approximate the hypergeometric distribution by the sum of i.i.d. Bernoulli random variables and apply the Central Limit Theorem to this sum (we use the so-called Berry–Esseen version of the CLT, which provides an upper bound on the error term). As a result, we get a statement even stronger than the lower bound~\eqref{eq_k_MAJ_lower}: this lower bound holds for $\AR_{n,m}\big(\kMAJ{k}\big)$ if that $n$ is big enough compared to $k$, namely,  $2(k+1)^2\leq n$. \rmr{we can add this to the theorem statement. Cleaner I think}\fsy{Done}
\end{proof}
\fi

\ifdefined\EX
\begin{corollary}[Asymptotic behavior]\label{cor_asymp_MAJ}
Combining upper and lower bound we see that $\AR_m\big(\kMAJ{k}\big)=1+\Theta\left(\frac{1}{\sqrt{k}}\right)$. Proof of the theorem demonstrates that  worst-case preference profiles for large $k$ are those, where both alternatives have approximately equal support with the imbalance created by $\Theta\left(\frac{1}{\sqrt{k}}\right)$-fraction of voters.\rmr{worth discussing that in the hard region of CJT (close to $0.5$), the gain from additional voters is also square root - success probability is something like $0.5+\Theta(\sqrt{k})$} \fsy{Could you phrase it?}
\end{corollary}
\begin{corollary}[Deterministic committees]
By the probabilistic-method argument, the theorem implies the existence of a deterministic committee of size $k$ such that the majority vote of its members has the approximation ratio at most $1+\frac{6\exp\left(-\frac{1}{2}\right)}{\sqrt{k}}$.
\end{corollary}
\begin{remark}[Optimal size of the committee]\label{cor_opt_size}\rmr{I suggest taking this discussion out of the corollary environment}\fsy{Will remark be better? I thought we need some environment to clearly separate thus example from the rest of the paper since we are discussing here an enriched model with costs.}
Assume that eliciting preference of one voter on one issue has some fixed cost $c$ measured in the same units as social costs of a voting outcome. Let's determine the optimal committee size that minimizes the worst-case regret for large number of voters $n$. Regret is defined as the difference between the total cost of $\kMAJ{k}$ (including the elicitation cost) and the social cost of the optimal outcome
$$\REG(X)=\left[\SC\big(\kMAJ{k}(X)\big)+ c\cdot k \cdot m\right]- \SC(z^\opt(X)).$$
By Corollary~\ref{cor_asymp_MAJ},
$$\max_X \REG(X)=m \cdot n \cdot \Theta\left(\frac{1}{\sqrt{k}}\right)+c\cdot k \cdot m.$$
Minimizing the worst-case regret over $k$ is, therefore, equivalent to minimizing $\Theta\left(\frac{1}{\sqrt{k}}\right)+\frac{c}{n}\cdot k$.\footnote{A fine issue here is that the optimal social cost in the worst case increases linearly in $n$ and $m$. This follows from Lemma~\ref{lm_one_issue_gives_an_upper_bound} (for $m$) and from the fact that the worst instances are those where the minority fraction $p$ is far from $0$.}
\end{remark}
\begin{corollary}
 The optimal committee size is $k=\Theta\left(\left(\frac{n}{c}\right)^{\frac{2}{3}}\right).$
\end{corollary}
\rmr{it is a bit strange that here we use the difference rather than the ratio. Also you assume here that OPT scales linearly with $m$ and $n$ which is not obvious (although I see why it is true).}

 
\fi

\section{Delegation is good with few issues but harmful with many}\label{sec_delegation}

 In the previous section, we saw that the $k$-sortition achieves an approximation ratio close to $1$. A natural question to ask is whether we can push the approximation ratio even closer to $1$ by using slightly more information about preferences of the population.
 
 We first show that the approximation ratios from the previous section can be significantly improved if the number of issues $m$ is small. For this purpose, we use weighted majority voting to aggregate preferences of the committee members, where the weight of a member is determined by the fraction of voters to which this member is the ``closest'' representative~\cite{alger2006voting,cohensius2017proxy}. 

Similar approach to weighing representatives is studied in the field of proxy voting (see Section~\ref{sec_intro}). The common wisdom suggested by that strain of literature is that usually the quality of opinion-aggregation can be significantly improved by finding a clever way of transforming proximity to weights.
It turns out, that in case of large number of issues, this is not the case.

While for small number of issues, weighing makes the approximation ratio exponentially close to $1$ in the committee size $k$ (compare with $1+O(1/\sqrt{k})$ from Theorem~\ref{th_approx_ratio_arbitrary_k}), we show that for large number $m$ of issues,  delegation is harmful and the approximation ratio is bounded from below by $\frac{9}{8}$ for any size of the committee $k$.

For a committee $C\subset [n]$ and  a voter $i\in [n]$, denote by $C_i$ the subset of committee members closest to $i$, his best representatives. We assume that each voter distributes one unit of ``weight'' over his closest representatives. So the weight $w_c$ of a committee member $c\in C$ is given by  $w_c=\sum_{i\in [n]:\ c\in {C}_i} \frac{1}{|{C}_i|}.$ 
\begin{definition}[\new{$\kREP{k}$: weighted $k$-sortition}] A committee $C\subset [n]$ of size $|C|=k$ is selected uniformly at random. 

The committee decides on each issue $j\in[m]$ using weighted majority vote with weights $w_c$: if $\sum_{c\in C} x_{c,j}\cdot w_c>\frac{1}{2}\sum_{c \in C} w_c$, then the outcome is $z_j=1$; in case of the the opposite strict inequality, $z_j=0$; in case of a tie, $z_j\in\{0,1\}$ is taken equally likely.
\end{definition}
\rmr{It is not clear how this is a committee voting rule according to Def.~\ref{def:cvr}. If $f$ is weighted majority, then it relies on information outside $C$, and if $f$ is majority, then it is not clear what is the distribution over subsets, as it is not defined explicitly. We should either give an explicit construction or extend the definition.}\fsy{Added the discussion and example in the Sect 2}

\ifdefined\EX
The following simple example shows that the approximation ratio is exponentially close to $1$ in the one-issue case.
\begin{example}[$\kREP{k}$ for $1$ issue]\label{ex_exp_improvement}
Assume that $n_1$ among $n$ voters support the alternative $0$ and denote $\frac{n_1}{n}$ by $p$. Without loss of generality, $0<p<\frac{1}{2}$, i.e., the majority of voters prefers the alternative $0$ that has the social cost of $p\cdot n$.

If the preferences of all the committee members coincide with those of majority ($x_{c,1}=0,\ \forall c\in C$), then, the committee selects the optimal alternative $0$. Similarly, if $x_{c,1}=1$ for all  $c\in C$, the sub-optimal alternative $1$ with the social cost $(1-p)\cdot n$ is selected.

Consider the remaining case: both views $0$ and $1$ are presented in $C$. We check that, in this case, the decision is socially optimal for any proportion of supporters of each of the alternatives (compare with at least $\frac{k+1}{2}$ members needed without weighing!). Indeed, the total weight received by the supporters of zero alternative is $(1-p)\cdot n$ and is bigger than the weight of the supporters of one, $p\cdot n$.

Thus the committee makes the sub-optimal decision if and only if it contains no supporters of zero alternative, which happens with probability ${\footnotesize {\begin{pmatrix}n_1\\k\end{pmatrix}}/{\begin{pmatrix}n\\k\end{pmatrix}}}\leq p^k$. Therefore, the approximation ratio admits the following upper bound 
$$\AR_{n,1}\big(\kREP{k}\big)\leq \max_{p\in\left(0,\frac{1}{2}\right)}\frac{(1-p^k)\cdot p\cdot n + p^k\cdot  (1-p)\cdot n}{p\cdot n}=1+ \max_{p\in\left(0,\frac{1}{2}\right)} p^{k-1}(1-2p).$$
To find the worst-case, we take the logarithmic derivative and get $\frac{k-1}{p}=\frac{1}{\frac{1}{2}-p}$; hence, $p^*=\frac{1}{2}\frac{k-1}{k}=\frac{1}{2}-\frac{1}{2k}$. We get
$$\AR_{n,1}\big(\kREP{k}\big)\leq 1+\frac{1}{k\cdot 2^{k-1}}\left(1-\frac{1}{k}\right)^{k-1}.$$
By taking a sequence of preference profiles with growing number $n$ of voters and $\lfloor p\cdot n\rfloor$ supporters of the alternative $1$, we see that the upper bound is achieved in the limit. Thus
$$\AR_{m=1}\big(\kREP{k}\big)=1+\frac{1}{k\cdot 2^{k-1}}\left(1-\frac{1}{k}\right)^{k-1}=1+\frac{2}{e}\cdot \frac{1}{k\cdot 2^k}(1+o(1)),\quad k\to \infty.$$
We see that weighing drastically improves the approximation ratio: from $1+O\left(\frac{1}{\sqrt{k}}\right)$ of Theorem~\ref{th_approx_ratio_arbitrary_k} to $1+o\left(2^{-k}\right)$.
\end{example}
\fi
If there are several issues, weights of the committee members depend on all of them. Therefore we don't have an analog of Lemma~\ref{lm_one_issue_gives_an_upper_bound}, which allowed us to analyze the rule $\kMAJ{k}$ in the one-issue case and then extend the results in a straightforward way. In particular, for $\kREP{k}$, the approximation ratio depends on the number of issues.

Now we demonstrate that for small number of issues, delegation leads to exponential improvement as in Example~\ref{ex_exp_improvement}. To simplify the analysis we focus on preference-profiles with i.i.d. issues. 

We say that the profile $X=(x_{i,j})_{i\in [n],j\in[m]}$ has \emph{i.i.d. issues} if, when we sample a voter $i\in[n]$ uniformly at random, the random variables $(x_{i,j})_{j\in m}$ are independent and identically distributed. Informally this property means that issues are similar but unrelated: there are no logical dependencies between them or other correlations in preferences. 

We define approximation ratio for i.i.d. issues by restricting the maximization in~\eqref{eq_def_AR} to profiles with i.i.d. condition
\begin{equation*}
\AR_{n,m}^\iid(f)=\max_{\mbox{{\footnotesize i.i.d.}}\,  X\in \{0,1\}^{n\times m}} \left(\frac{\SC\big(f(X)\big)}{\SC\big(z^{\opt}\big)}\right).   
\end{equation*}
Note that $\AR_{n,m}^\iid$ is always below $ \AR_{n,m}$.
\begin{proposition}\label{prop_delegation_small_m}
The approximation ratio of $\kREP{k}$ satisfies the following upper bound for profiles with  $m$ i.i.d. issues
$$\AR_{n, m}^\iid\big(\kREP{k}\big) \leq 1+m^{m+1}\cdot\exp\left(-\frac{k}{(2m)^m}\right)$$
\end{proposition}
\fsy{We can probably assume independence instead of i.i.d.}
We see that for small number $m$ of issues the approximation ratio converges to $1$ exponentially fast in the size of the committee. Thus, as in Example~\ref{ex_exp_improvement}, the delegation improves the asymptotic behavior compared to $\kMAJ{k}$ (note that $\AR_{n,m}^\iid$ and $\AR_{n,m}$ coincide for $\kMAJ{k}$; the argument is similar to Lemma~\ref{lm_one_issue_gives_an_upper_bound} and is omitted). However, the rate of exponential convergence drops drastically with the growth of the number of issues $m$. Below we will see that this is not a coincidence.

Proposition~\ref{prop_delegation_small_m} is proved in Appendix~\ref{app_proofs_delegation} and the proof is sketched here.
\ifdefined\SKETCH
\begin{proof}[Proof sketch of Proposition~\ref{prop_delegation_small_m}]
It is enough to look at i.i.d. profiles $X$, where the majority prefers $0$ for each issue; we  denote by $p\leq \frac{1}{2}$ the fraction of supporters of the alternative $1$.

If $p$ is very small ($p< \frac{1}{2m}$), then by the union bound, more than half of the population prefers $0$ for all issues. Therefore, if at least one ``all 0" voter enters the committee, they get more than half of the total weight  and the committee selects the socially optimal outcome. Since such  voters are likely to be among the committee members (the probability is at least $1-(pm)^k$), the social cost of $\kREP{k}$ at $X$ is exponentially close to optimal.\rmr{if we extend to non iid, then the issues with $p_j<1/2m$ contribute together less than $(pm)^k$ to the error.}

For $p\geq \frac{1}{2m}$, the argument for the low social cost is different. Using the union bound we show that, with high probability, for each sequence $z\in\{0,1\}^m$ there is a committee member with such preferences. Note that given this event, each voter delegates his voting right to a member with exactly the same preferences as his own, and hence, the outcome of the committee vote coincides with the majority vote of the whole population. Details of the computation can be found in  Appendix~\ref{app_proofs_delegation}.\rmr{in the non-iid case we have the problem that while all voters have a perfect representative for the ``hard" issues, he may not represent them on the easy issues since he is not ``all 0" on them (and only half of them are ``all 0" anyway).}
\end{proof}
\fi
The next result complements Proposition~\ref{prop_delegation_small_m} and shows that for unbounded number of issues, delegation is harmful and the approximation ratio is separated from $1$ even for big committees.
\begin{proposition}\label{prop_delegation_is_bad}
For unbounded number of issues, the approximation ratio of $\kREP{k}$ admits the following lower bound
$$\AR\big(\kREP{k}\big)\geq \frac{9}{8}$$
for any committee size $k$. 
\end{proposition}
\begin{corollary}
Comparing this result with Theorem~\ref{th_approx_ratio_arbitrary_k}, we see that 
$\kMAJ{1000}$ outperforms $\kREP{k}$ for any committee size $k$. 
\new{Note that ancient Athenian democracy had $1100$ magistrates, $1000$ of whom were selected by lot~\citep{hansen1999athenian}. Committees of size $1000$ are also recommended by \citet{mueller1972}.}
\rmr{maybe say $\kMAJ{1000}$? that is what recommended in \cite{mueller1972}}\fsy{Done} 
\end{corollary}
The fact that delegation in multi-issue setting may lead to inferior outcomes compared to majority vote is mentioned in~\citep[Section 5]{cohensius2017proxy} for a similar model but without the worst-case analysis. Their insight is that for some preference profiles delegation leads to the most extreme voters attracting almost all weight, thus wasting the information about others' preferences. We build upon this insight.
\ifdefined\SKETCH
\begin{proof}[Proof sketch (for details, see Appendix~\ref{app_proofs_delegation})]
The idea is to construct a preference profile $X$ with two groups of voters $P$ and $Q$. All voters from $Q$ have identical preferences and all voters from $P$ are equidistant and the distance between any two distinct voters from $P$ is bigger than the distance between a voter from $P$ and a voter from $Q$.

Given such a metric structure, if any voters from $Q$ enter the committee, they receive the weight both from $Q$ and $P$ and thus the committee decision becomes dictatorial and coincides with preferences of a voter from $Q$. In particular, the social cost of the outcome coincides with the social cost of the unanimous preferences of $Q$.

In the appendix, we show existence of such a preference profile $X$, where the social cost of $Q$'s preferences is $\frac{9}{8}$ of the optimum.
The profile $X$ is constructed via the probabilistic argument: $(x_{i,j})_{i\in P, j\in [m]}$ are i.i.d. Bernoulli random variables with success probability $p < \frac{1}{2}$;  and for $i\in Q$ and $j\in[m]$ we put $x_{i, j}=\xi_j$, where $(\xi_j)_{j\in [m]}$ are i.i.d. Bernoulli random variables with success probability $q<p$. Therefore, for large number of issues $m$, the distance between any two voters from $P$ is approximately $2p(1-p)m\cdot(1+o(1))$, while the distance between a voter from $P$ and a voter from $Q$ is $(p(1-q)+q(1-p))m\cdot(1+o(1))$. Optimization over $p$, $q$, and sizes of $P$ and $Q$ leads to the desired lower bound.
\end{proof}
\fi

\section{Optimality of Sortition}\label{sec_general}
In Section~\ref{sec_delegation}, we considered weighing of the committee members, where a voter delegates one unit of weight to the closest representative. This method proved to be much better than the simple majority rule for a few i.i.d. issues, however it gives no advantage and may even harm if the number of issues is big.

One may think that poor guarantees for large number of issues is a feature of this particular delegation method. Indeed, it is easy to come up with alternative proposals that seem to be better:
a voter may ``smoothly'' distribute the weight among the committee members in a way that members that are closer to him get more weight. Also, instead of selecting the committee uniformly, a higher chance can be given to voters with smallest average distance to the rest of the population. 

In this section we show that none of such natural modifications can improve the approximation ratio with many issues:  $\kMAJ{k}$, a uniformly random committee with the simple majority rule, is worst-case optimal among a large family of rules that may use the whole metric information about the preference profile.

\rmr{I think every ``neutral" (invariant to permutations of the coordinates) and anonymous rule must be of this form, no? } \rmr{also, maybe we can tie it to strategyproofness. If we show optimality in the set of all SP rules this is great. $\kREP{k}$ is not SP I think:  suppose there are 1 ``critical issue" $X$ and 10 other dummy issues $Z$. agent~$i$ is all 0. there are 100 voters that vote 0 on $X$ and 1 on $Z$, 100 voters that vote 1 on $X$ and 1 on $Z$. There are also 10 $j$ voters who vote 1 on $X$ and 0 on $Z$, and 10 $k$ voters who vote 0 on $k$ and 1 on half of $Z$ issues.

So the chance that $i$ will be in the committee is low, but it is likely to contain both $j$ and $k$ voters. Also, dummy issues $Z$ will probably be all 1 regardless of $i$'s vote. However the critical $X$ issue is likely to be close to a tie. by being truthful. $i$ is closer to $j$, but by lying it can become closer to $k$ and affect the critical issue.}

%
\begin{theorem}\label{theorem_general}
For any distance-based $k$-committee rule $g$ with issue-wise inter-committee vote and any number of voters $n$, we have 
$$\sup_{m\in \N} \AR_{n,m}\big(g\big)\geq \AR_{n,1}\big(\kMAJ{k}\big).$$
In other words, a uniformly-random committee with the simple majority rule is worst-case optimal within this family of committee rules. 
\end{theorem}
We prove a stronger result: instead of supremum over $m$ on the left-hand side, one can take $m=n!$ issues.

The proof is based on two lemmas. The first one shows that it is enough to prove the theorem for anonymous committee rules, i.e., symmetric with respect to permutation of voters. 

For a preference profile $X$ with $n$ voters and a permutation $\pi$ of $[n]$, we denote by $\pi(X)$ the preference profile with  permuted voters: $\pi(X)_{i,j}=X_{\pi(i),j}$.  
\begin{definition}[anonymous voting rules]
A voting rule $h$ is anonymous if $h(X)=h\big(\pi(X)\big)$ for any preference profile $X$ and permutation $\pi$.
\end{definition}
Note that even if  committee members receive different weights as in Section~\ref{sec_delegation} (and thus inter-committee rule is not anonymous), the committee rule 
may still be anonymous if the selection of the committee and the inter-committee rule depends on preferences of the population in an anonymous way. For example, $\kREP{k}$ rule is anonymous.

For a rule $h$, denote by $h^{\sym}$ its \emph{anonymization}: for any preference profile~$X$
$$h^{\sym}(X)=\frac{1}{n!}\sum_{\pi\in S_n}h\big(\pi(X)\big),$$
where $n$ is the number of voters and $S_n$ is the set of all permutations of $[n]$.  
\begin{lemma}\label{lm_anonymization}
For any numbers of voters and issues and any voting rule $h$
$$ \AR_{n,m}\big(h^\sym\big)\leq \AR_{n,m}\big(h\big).$$
\end{lemma}
The proof is elementary and can be found in Appendix~\ref{app_proofs_general}.

The combination of anonymity and the distant-based property becomes very restrictive if distances between pairs of voters are all the same. The following lemma demonstrates that the set of such preference profiles is rich enough.
\begin{lemma}\label{lm_equidistant}
For any number of voters $n$ and $n_1\leq n$, there exists a preference profile with $n$ voters and $m=n!$ issues such that for each issue $j\in [m]$ exactly $n_1$ voters support the alternative $1$ and for each pair of distinct voters $i,i'\in[n]$ the distance $d(x_i,x_{i'})=\frac{2p(1-p)}{1-\frac{1}{n}}$, where $p=\frac{n_1}{n}$.
\end{lemma}
\ifdefined\SKETCH
\begin{proof}[Sketch of the proof (see Appendix~\ref{app_proofs_general} for the formal argument):]
The idea is similar to the one used in Proposition~\ref{prop_delegation_is_bad}: if preferences of each voter are given by a vector of i.i.d. Bernoulli random variables with success probability $p=\frac{n_1}{n}$, then, in expectation, $n_1$ voters prefer alternative $1$ for each $j\in[m]$, and the distance between any pair of voters concentrates around $2p(1-p)m$ for large $m$ by the law of large numbers.

To prove Lemma~\ref{lm_equidistant} we use a derandomization of this construction. This allows to have exactly $n_1$ supporters of $1$ for each issue and exact equality of distances. Details can be found in Appendix~\ref{app_proofs_general}. 
\end{proof}

Now we are ready to prove the theorem.\fi
\begin{proof}[Sketch of the proof of Theorem~\ref{theorem_general} (see Appendix~\ref{app_proofs_general} for details)]
By Lemma~\ref{lm_anonymization}, we can assume that $g$ is anonymous without loss of generality. 

Using Lemma~\ref{lm_equidistant}, we pick a profile $X$ such that all the
voters are equidistant. On such a profile, the metric information encoded in the distance matrix becomes useless and the only anonymous way to select a committee is uniformly at random.

The next step is to show that among the family of all committee rules with uniformly random committee and ble inter-committee vote, no rule can outperform $\kMAJ{k}$ both on the profile $X$ and on the complement profile $\bar{X}$, where preference of all the voters are flipped. This boils down to solving an explicit optimization problem; the details can be found in Appendix~\ref{app_proofs_general}.

Lemma~\ref{lm_equidistant} provides enough flexibility to make $X$ (and, hence, $\bar{X}$), the worst-case profile for $\kMAJ{k}$. Since $g$ is worse off on one of these profiles, it has higher approximation ratio.
\end{proof}

\section{Discussion}\label{sec_conclusions}

The main takeaway message from our work is that at least when the preferences of the individuals in the society are separable, the simple and well known $k$-sortition rule is a reasonable choice that guarantees low approximation ratio.

While for a low number of issues we saw that there are more representative rules than the $k$-sortition rule, the latter also has other benefits. 
First, the selection of the committee $C$ (and its internal voting rule $h$) is completely oblivious to voters' positions. As such, it can be used to select a committee even before we know some or all of the issues on the agenda. 

Second, while we did not focus on incentive analysis in this work, it is easy to see that $k$-sortition is strategyproof for the the entire population: no voter can affect the selection of the committee, and representatives are always weakly better of by voting their true position on every issue. In contrast, under the proxy-weighted variant, voters may have an incentive to lie in order to affect the weights of committee members. 

\medskip
The main direction we are currently investigating is a better understanding of the effects of delegation, in particular when there are few issues and/or our restrictions on the structure of preferences. 

In the long run, we are interested in how sortition and/or delegation-based voting rule can guarantee low distortion and good approximation in other domains, including ranked preferences and interdependent issues.



\nocite{}
\addcontentsline{toc}{section}{\protect\numberline{}References}%

\bibliographystyle{plainnat}

\bibliography{references}

\appendix

\section{Proofs for Section~\ref{sect_uniform_majority}}\label{app_proofs_unif_maj}

\begin{proof}[Proof of Lemma~\ref{lm_one_issue_gives_an_upper_bound}]
In order to show that $\AR_{n,m}\big(\kMAJ{k}\big)\leq \AR_{n,1}\big(\kMAJ{k}\big)$, consider a worst-case profile $X^*\in \{0,1\}^{n\times 1}$ for one issue. Define a new profile $X$ with $m$ issues by cloning $X^*$: $x_{i,j}=x^*_i$. Since different issues contribute to the Hamming distance in an additive way and $\kMAJ{k}$ operates on each issue separately,
\begin{equation}\label{eq_lm1_proof_1}
\AR_{n,m}\big(\kMAJ{k}\big)\geq \frac{\SC\big(\kMAJ{k}(X)\big)}{\SC\big(z^\opt(X) \big)}= \frac{m\cdot \SC\big(\kMAJ{k}(X^*)\big)}{m\cdot\SC\big(z^\opt(X^*) \big)}=\AR_{n,1}\big(\kMAJ{k}\big).
\end{equation}
To prove the opposite inequality, pick a worst-case profile $Y^*\in \{0,1\}^{n\times m}$ with $m$ issues and denote by $Y^j$ a one-issue profile obtained by the restriction of $Y$ to issue $j$, i.e., $y^j_i=y^*_{i,j}$. 
The following computation shows that the worst-case profile with $m$ issues is not worse than its restriction to the worst issue
$$\AR_{n,m}\big(\kMAJ{k}\big)=\frac{\SC\big(\kMAJ{k}(Y^*)\big)}{\SC\big(z^\opt(Y^*) \big)}=\frac{\sum_{j\in [m]}\SC\big(\kMAJ{k}(Y^j)\big)}{\sum_{j\in [m]}\SC\big(z^\opt(Y^j) \big)}\leq$$
\begin{equation}\label{eq_lm1_proof_2}
\leq \max_{j\in [m]}\frac{\SC\big(\kMAJ{k}(Y^j)\big)}{\SC\big(z^\opt(Y^j) \big)}\leq \AR_{n,1}\big(\kMAJ{k}\big).
\end{equation}
Combining inequalities~\eqref{eq_lm1_proof_1} and~\eqref{eq_lm1_proof_2}, we obtain the desired equality of approximation ratios.
\end{proof}

\begin{proof}[Proof of Theorem~\ref{th_approx_ratio_arbitrary_k}]
By Lemma~\ref{lm_one_issue_gives_an_upper_bound}, it is enough to consider a one-issue profile $X$. 
Without loss of generality, we assume that the majority supports the alternative $0$:   $1\leq n_1\leq \frac{n}{2}$ voters prefer the alternative $1$, and the rest prefer $0$. Denote $p=\frac{n_1}{n}\in \left(0,\frac{1}{2}\right]$. 

Let $\xi$ be the number of committee members supporting the alternative $1$. The random variable $\xi$ has the hypergeometric distribution with parameters $(n,n_1=pn,k)$. Recall that the hypergeometric distribution describes the number of red balls among $k$ draws without replacement from an urn with $n_1$ red and $n-n_1$ white balls. If we represent supporters of $0$ as white balls and supporters of $1$ as red and take into account that sampling without replacement is equivalent to picking a uniformly random subset, it becomes obvious that $\xi$ is hypergeometric.

If $\xi<\frac{k}{2}$, then the committee selects the socially-optimal alternative $0$ with the social cost $p\cdot n$; for $\xi=\frac{k}{2}$ (can happen only if $k$ is even), there is a tie and the committee picks any of the two alternatives equally likely; for $\xi>\frac{k}{2}$, the committee selects the sub-optimal alternative $1$ with the social cost $(1-p)\cdot n$. Therefore, the ratio of social costs  in~\eqref{eq_def_AR} can be represented as
\begin{equation}\label{eq_AR_representation}
    \frac{\SC\big(\kMAJ{k}(X)\big)}{\SC\big(z^\opt\big)}= \frac{\P\left(\xi<\frac{k}{2}\right)\cdot p\cdot n + \P\left(\xi=\frac{k}{2}\right)\cdot \frac{n}{2}+ \P\left(\xi>\frac{k}{2}\right) (1-p)\cdot n}{p\cdot n}\leq
\end{equation}
$$ \leq     1+\P\left(\xi\geq \frac{k}{2}\right)\frac{1-2p}{p}.$$
 We will consider the two cases of $p$ close to zero and $p$ close to $\frac{1}{2}$ separately. As we will see, the worst-case instances correspond to the latter scenario, however, $p$ in the denominator of \eqref{eq_AR_representation} does not allow to cover both cases at once.
 
  
 \paragraph{The case $p\leq \frac{1}{6}$.} For small $p$, we use a rough bound on $\P(\xi\geq \frac{k}{2})$ based on the Chebyshev inequality.

The expectation and the variance of a hypergeometric random variable are given by 
$$\E \xi=p\cdot k\ \  \ \mbox{and}  \ \ \  \V \xi= p(1-p)\frac{n-k}{n-1}\cdot k\leq p(1-p)\cdot k,$$
see e.g. \citep[\S 2.6]{feller2008introduction}.\fsy{check the reference}
By the Chebyshev inequality,
$$\P\left(\xi\geq\frac{k}{2}\right)\leq \frac{\V \xi}{\left(\frac{k}{2}-p\cdot k\cdot \E \xi\right)^2}\leq \frac{4p(1-p)}{(1-2p)^2}\frac{1}{k}.$$
Thus by~\eqref{eq_AR_representation},
\begin{equation}\label{eq_AR_p_close_to_0}
\frac{\SC\big(\kMAJ{k}(X)\big)}{\SC\big(z^\opt\big)}\leq 1+\frac{4(1-p)}{1-2p}\frac{1}{k}\leq1+\frac{5}{k}. 
\end{equation}

\paragraph{The case $\frac{1}{6}\leq p\leq \frac{1}{2}$.} To estimate the probability of $\xi\geq \frac{k}{2}$, we use the tail bounds for hypergeometric random variables. They enjoy the usual Hoeffding inequality
\begin{equation}\label{eq_Hoeffding}
\P\left(\xi\geq \E \xi+t\right)\leq \exp\left(-2\cdot \frac{t^2}{k} \right) \ \ \mbox{for} \ \ t\geq 0
\end{equation}
as if $\xi$ was the sum of $k$ Bernoulli random variables with success probability $p$.
The inequality for hypergeometric distribution is proved in the same paper of \citet{hoeffding1994probability}, see Section~5 there. 

From~\eqref{eq_Hoeffding}, we get $\P\left(\xi\geq \frac{k}{2}\right)\leq \exp\left(-\frac{1}{2}(1-2p)^2\cdot k\right)$. Substituting this bound in~\eqref{eq_AR_representation}, replacing $p$ in denominator by its lower-bound $\frac{1}{6}$, and  denoting $1-2p$ by $t$, we obtain
$$\frac{\SC\big(\kMAJ{k}(X)\big)}{\SC\big(z^\opt\big)}\leq 1+\exp\left(-\frac{1}{2}(1-2p)^2\cdot k\right)\frac{1-2p}{p}\leq 1+ 6\max_{t\in\left[0,\frac{1}{2}\right]} \exp\left(-t^2\cdot k\right)\cdot t.$$
Logarithmic derivative $-t\cdot k+\frac{1}{t}$ is zero at $t=\frac{1}{\sqrt{k}}$. 

Hence,
\begin{equation}\label{eq_AR_p_close_to_0.5}
\frac{\SC\big(\kMAJ{k}(X)\big)}{\SC\big(z^\opt\big)}\leq 1+ \frac{6\exp\left(-\frac{1}{2}\right)}{\sqrt{k}}.
\end{equation}

\paragraph{Gluing the two cases.} For $k\geq 2$, the upper bound~\eqref{eq_AR_p_close_to_0.5} proved for $\frac{1}{6}\leq p\leq \frac{1}{2}$ exceeds~\eqref{eq_AR_p_close_to_0} derived for $p\leq \frac{1}{6}$. Indeed, 
$$\frac{5}{k}\leq \frac{6\exp\left(-\frac{1}{2}\right)}{\sqrt{k}}\Longleftrightarrow k\geq \frac{25}{36}\cdot e\approx 1.9.$$
Therefore,~\eqref{eq_AR_p_close_to_0.5} bounds the  ratio of the social costs for any $p$ and $k\geq 2$.  For $k=1$, the right-hand side of~\eqref{eq_AR_p_close_to_0.5} is equal to $1+6\exp\left(-\frac{1}{2}\right)\approx 4.6$, while the worst-case approximation ratio for $\kMAJ{1}$ (a random dictator) is below $2$ by Example~\ref{ex_RD}. Thus the approximation ratio satisfies 
$$\AR_{n,1}\big(\kMAJ{k}\big)=\max_X \frac{\SC\big(\kMAJ{k}(X)\big)}{\SC\big(z^\opt\big)}\leq 1+ \frac{6\exp\left(-\frac{1}{2}\right)}{\sqrt{k}} $$
for all $n$ and $k$.

\paragraph{Asymptotic tightness.}
 We will consider the scenario where $2 (k+1)^2\leq n$, i.e., the number of voters $n$ is large compared to the committee size $k$.
 
From~\eqref{eq_AR_representation}, we deduce the lower bound: 
\begin{equation}\label{eq_AR_lower_bound}
    \frac{\SC\big(\kMAJ{k}(X)\big)}{\SC\big(z^\opt\big)}\geq 1+\P\left(\xi>\frac{k}{2}\right)\frac{1-2p}{p}.
\end{equation}
Let's estimate $\P\left(\xi>\frac{k}{2}\right)$ from below for the hypergeometric distribution with parameters $(n,n_1=pn,k)$. If $k'\leq k $ balls are already taken and $n_1'\leq k'$ of them are red, the chance to pick the next red ball is $\frac{n_1-n_1'}{n-k'}\geq \frac{n_1-k}{n}=p-\frac{k}{n}$. Therefore, $\P\left(\xi>\frac{k}{2}\right)$ is bounded from below by the probability $\P\left(\sum_{i=1}^k \eta_i>\frac{k}{2}\right)$, where $\eta_i$ are i.i.d. Bernoulli random variables with success probability $p-\frac{k}{n}$. 

For i.i.d. random variables, the probability $\P\left(\sum_{i=1}^k \eta_i\leq t\right)$ can be approximated by the normal law using the Berry–Esseen theorem (a version of the central limit theorem with a bound on approximation error \citep{shevtsova2011absolute})
$$\left|\P\left(\sum_{i=1}^k \eta_i \leq t\right)- \Phi\left(\sqrt{k}\cdot\frac{\frac{t}{k}-\E \eta_1}{\sqrt{\V \eta_1}}\right)\right|\leq \frac{1}{2}\frac{\E \big|\eta_1-\E \eta_1\big|^3 }{\left(\V \eta_1\right)^\frac{3}{2}}\leq \frac{1}{2\sqrt{k}},$$
where $\Phi(z)=\frac{1}{\sqrt{2\pi}}\int_{-\infty}^z \exp\left(-\frac{y^2}{2}\right)\, dy$.
Thus 
$$\P\left(\xi>\frac{k}{2}\right)\geq \P\left(\sum_{i=1}^k \eta_i\leq t\right)\geq 1-\Phi\left(\sqrt{k}\left(1-2p+\frac{2k}{n}\right)\right)-\frac{1}{2\sqrt{k}},$$
where we wrote $\frac{1}{4}$ instead of $\V \eta_1\leq \frac{1}{4}$.

Let the number of supporters for the suboptimal alternative be $$n_1=\Big\lceil\left(\frac{1}{2}-\frac{1}{2\sqrt{k}}+\frac{k}{n}\right)\cdot n\Big\rceil$$ 
and hence $\frac{1}{2}-\frac{1}{2\sqrt{k}}+\frac{k}{n}\leq p<\frac{1}{2}-\frac{1}{2\sqrt{k}}+\frac{k}{n}+\frac{1}{n}.$ Note that $p\in\left(0,\frac{1}{2}\right)$ by the assumption  on $n$ and $k$. We get
$$\P\left(\xi>\frac{k}{2}\right)\geq 1-\Phi(1)-\frac{1}{2\sqrt{k}}=\Phi(-1)-\frac{1}{2\sqrt{k}}.$$  
Substituting this into~\eqref{eq_AR_lower_bound}, we obtain
\begin{eqnarray}
\notag \frac{\SC\big(\kMAJ{k}(X)\big)}{\SC\big(z^\opt\big)}\geq& 1+2\left(\Phi(-1)-\frac{1}{2\sqrt{k}}\right)\left(\frac{1}{\sqrt{k}}-\frac{2(k+1)}{n}\right)\geq\\
\notag \geq& 1+\frac{2\Phi(-1)}{\sqrt{k}}-\frac{1}{k}-{2\Phi(-1)}\cdot\frac{2(k+1)}{n}.
\end{eqnarray}
Inequalities $\Phi(-1)\leq \frac{1}{2} $ and $\frac{2(k+1)}{n}\leq \frac{1}{k}$ lead to the lower bound on the approximation ratio
$$\AR_{n,1}\big(\kMAJ{k}\big)=\max_X \frac{\SC\big(\kMAJ{k}(X)\big)}{\SC\big(z^\opt\big)}\geq 1+\frac{2\Phi(-1)}{\sqrt{k}}-\frac{2}{k}$$
valid for $2(k+1)^2\leq n$. 
\end{proof}

\section{Proofs for Section~\ref{sec_delegation}}\label{app_proofs_delegation}

\begin{proof}[Proof of Proposition~\ref{prop_delegation_small_m}]
Consider a preference profile $X$ with i.i.d. issues. Without loss of generality, for every issue $j\in [m]$, the majority of voters prefers the alternative $0$ but not unanimously. We denote by $p$ the fraction of supporters of the alternative $1$, so $0< p\leq \frac{1}{2}$.

Consider first the case of small $p<\frac{1}{2m}$, which will not require the independence assumption. By the union bound, at least $1-pm>\frac{1}{2}$ fraction of all voters prefer the alternative $0$ on every issue. Therefore, if the committee $C$ contains a ``zero'' member  that prefers $0$ for every issue, he receives more than $1/2$ of the total weight and the alternative $0$ wins for all issues thus providing the socially optimal outcome.

The probability that a random committee contains a zero member is at least $1-(pm)^k$, therefore, the ratio of social costs satisfies
$$\frac{\SC\big(\kREP{k}\big)(X)}{\SC(z^\opt)}\leq\frac{\left(1-(pm)^k\right)\cdot p +(pm)^k\cdot (1-p)}{p}= $$
$$=1+ m\cdot (pm)^{k-1}(1-2p)\leq 1+ \frac{m}{2^{k-1}}.$$

Now consider the case of $\frac{1}{2m}\leq p\leq \frac{1}{2}$. Let's estimate that probability of the event $A$ that for every sequence  $z\in \{0,1\}^m$ there is a committee member $c\in C$ with $x_c=z$. Given $A$, every voter delegates his vote to a committee member with exactly the same preferences as his own and thus the outcome of the committee vote coincides with the outcome of the majority vote of the whole population. Thus, if $A$ occurs, the committee selects the socially-optimal outcome.

We estimate the probability of $A$ using the union bound. For any $z\in \{0,1\}^m$ with $q$ ones, the probability that there is no committee member with such preferences is is at most $\left(1-p^q(1-p)^{m-q}\right)^k$. Since there are ${\footnotesize\begin{pmatrix}m\\q\end{pmatrix}}$ such vectors $z$, we get
$$\P(A)\geq 1-\sum_{q=0}^m \begin{pmatrix}m\\q\end{pmatrix} \left(1-p^q(1-p)^{m-q}\right)^k.$$
Taking into account that  $\left(1-p^q(1-p)^{m-q}\right)^k\leq \left(1-p^m\right)^k$ and $\sum_{q=0}^m {\footnotesize\begin{pmatrix}m\\q\end{pmatrix}}=2^m$, we obtain
$$\P(A)\geq 1-2^m\left(1-p^m\right)^k.$$
Therefore
$$\frac{\SC\big(\kREP{k}\big)(X)}{\SC(z^\opt)}\leq \frac{\P(A)\cdot p + (1-\P(A))\cdot (1-p)}{p}\leq 1+ \frac{2^m}{p}\left(1-p^m\right)^k.$$
Since $\frac{1}{2m}\leq p$, the right-hand side  does not exceed
$$1+m\cdot 2^{m+1}\left(1-\frac{1}{(2m)^m}\right)^k\leq 1+m\cdot 2^{m+1}\exp\left(-\frac{k}{(2m)^m}\right),$$
where in the last inequality we used that $\left(1-t\right)^{\frac{1}{t}}\leq \frac{1}{e}$ for any $t\in(0,1)$.

Note that for any $m,k\geq 1$, this upper bound exceeds the one obtained for $p<\frac{1}{2m}$ and thus
$$\AR_{n,m}\big(\kREP{k}\big)\leq 1+m\cdot 2^{m+1}\exp\left(-\frac{k}{(2m)^m}\right).$$
\end{proof}

\begin{proof}[Proof of Proposition~\ref{prop_delegation_is_bad}]
Consider a random preference profile $X$ with
$n$ voters  and $m$ issues. There are two groups of voters $P$ and $Q$ with sizes $\alpha\cdot n$ and $(1-\alpha)\cdot n$, respectively. All the voters $i\in Q$ are unanimous: their preferences $x_i$ coincide with a vector $\xi=(\xi_j)_{j\in [m]}$ with components given by i.i.d. Bernoulli random variables with success probability $q<\frac{1}{2}$. In contrast, preferences of voters in $P$ are not aligned: $(x_{i,j})_{i\in P, j\in [m]}$ is a matrix with i.i.d. Bernoulli entries with success probability $p\in \left(q,\frac{1}{2}\right)$.

The expected distance between two distinct voters $i,i'\in P$ is equal to $\E [d(x_i,x_{i'})]=2p(1-p)m$. Similarly, for a voter $i\in P$ and $i'\in Q$ we have  $\E[d(x_i,x_{i'})]=(p(1-q)+q(1-p))m$. Expected distance to the zero vector $\E[d(x_i,0)]$ is equal to $pm$ for $i\in P$ and $qm$ for $i\in Q$.

By the law of large numbers, for any $\varepsilon>0$ we can find $m=m(\varepsilon,n,p,q)$ such that all the distances are within a factor $1\pm\varepsilon $ from their expected values with positive probability: 
$$(1-\varepsilon)\E[d(x_i,z)]\leq d(x_i,z)\leq (1+\varepsilon)\E[d(x_i,z)],\quad i\in [n],\ z\in \{x_{i'}\,:\, i'\in [n]\}\cup\{0\}.$$ 
Abusing the notation, we denote by $X$ the \emph{realization} of the random profile satisfying these inequalities.

Assuming that $\varepsilon$ is small, we see 
that the distance between any two voters from $P$ exceeds the distance between voters from $P$ and voters from $Q$ (it is enough to assume that $(1-\varepsilon)2p(1-p)m>(1+\varepsilon)(p(1-q)+q(1-p))m$).

Consider a random committee $C$ of size $k$.  If $C$ contains  some voters from $Q$, all voters from $[n]\setminus C$ delegate their weight to them. Therefore, if $n-k>\frac{k}{2}$ (i.e., $n$ is large enough compared to $k$), voters from $C\cap Q$ get more than half of the total weight and, therefore, their unanimous preferences $\xi$ coincide with the outcome of the committee vote. 

The chance that $C\cap Q$ is nonempty equals $1-{\tiny\frac{\begin{pmatrix}\alpha n\\ k\end{pmatrix}}{\begin{pmatrix}n\\ k\end{pmatrix}}}\geq 1-\alpha^k.$  Putting all the pieces together, we obtain the following lower bound on the approximation ratio
$$\AR\big(\kREP{k}\big)\geq\frac{(1-\alpha^k)\min_{i\in P, i'\in Q} d(x_i,x_{i'})\cdot \alpha n+ \alpha^k\cdot \SC(z^\opt) }{\SC(z^\opt)}.$$
Taking into account that $\SC(z^\opt)\leq \SC(0)$ and using the lower bound on the minimal distance, we get
$$\AR\big(\kREP{k}\big)\geq \alpha^k + (1-\alpha^k)\cdot \frac{(1-\varepsilon)\cdot \big(p(1-q)+q(1-p)\big)\cdot \alpha}{(1+\varepsilon)\cdot \big(\alpha p+(1-\alpha)q)\big)}.$$
Letting $\varepsilon\to +0$ and $q\to p-0$, we get rid of $\varepsilon$ and simplify the expression
$$\AR\big(\kREP{k}\big)\geq \alpha^k + (1-\alpha^k)\cdot 2(1-p)\alpha.$$
Now, letting $p$ go to $0$, we see that 
$$\AR\big(\kREP{k}\big)\geq 2\alpha+\alpha^k -2\alpha^{k+1}.$$
For $k=1$, the maximum of the right-hand side is achieved at $\alpha=\frac{3}{4}$. Substituting this value for any $k$, we obtain
$$\AR\big(\kREP{k}\big)\geq \frac{3}{2}-\frac{1}{2}\left(\frac{3}{4}\right)^k\geq \frac{9}{8}.$$
\end{proof}

\section{Proofs for Section~\ref{sec_general}}\label{app_proofs_general}

\begin{proof}[Proof for Lemma~\ref{lm_anonymization}]
Pick a profile $X$ with $n$ voters and $m$ issues and denote by $X^*$ the  profile of the form $\pi(X)$ maximizing $\SC\big(h(\pi(X))\big)$ over all permutations $\pi\in S_n$. Note that the socially-optimal cost is the same for $X$ and $X^*$. Therefore,
$$\frac{\SC\big(h^\sym(X)\big)}{\SC\big(z^\opt(X)\big)}=\frac{1}{n!}\sum_{\pi\in S_n} \frac{\SC\big(h(\pi(X))\big)}{\SC\big(z^\opt(X)\big)}\leq \frac{\SC\big(h(X^*)\big)}{\SC\big(z^\opt(X^*)\big)}\leq \AR_{n,m}\big(h\big).$$
By maximizing the left-hand side over $X$, we obtain the desired inequality.
\end{proof}

\begin{proof}[Proof of Lemma~\ref{lm_equidistant}]
Identify the set of issues with the set $S_n$ of all permutations  $\pi$ of  $[n]$. So there are $m=n!$ issues. Fix a set $A\subset [n]$ of cardinality $n_1$.
The set of voters preferring alternative $1$ for issue $\pi$ is defined to be the preimage $\pi^{-1}(A)$ of the set $A$ under permutation $\pi$. In other words,
preference of a voter $i$ on  an issue $\pi$ is given by
$$x_{i,\pi}=\left\{\begin{array}{cc}
    1, &  \pi(i)\in A \\
    0, &  \pi(i)\notin A.
\end{array}\right.$$
By the construction, exactly $n_1$ voter prefers $1$ for each issue. 

Consider the distance $d(x_i,x_{i'})=\sum_{\pi\in S_n }|x_{i,\pi}-x_{i',\pi}|$ between two distinct voters $i,i'\in[n]$. Let $\tau$ be the permutation such that $\tau(i)=1$
and $\tau(i')=2$. Taking into account that $\pi\to\pi\circ\tau$ defines an automorphism of $S_n$ and that $x_{i,\pi\circ\tau}=x_{\tau(i),\pi}$, we get 
$$d(x_i,x_{i'})=\sum_{{\pi\circ \tau}\in S_n }\big|x_{i,{\pi\circ\tau}}-x_{i',{\pi\circ\tau}}\big|=\sum_{{\pi}\in S_n }\big|x_{1,{\pi}}-x_{2,{\pi}}\big|=d(x_1,x_{2}).$$
Therefore, all the distances between distinct voters are equal to common constant $d=d(x_1,x_{2})$. To find $d$, note that if we take a pair of voters uniformly at random (not necessary distinct), the chance they have different opinion on a certain issue $\pi$ is $\frac{n_1(n-n_1)+(n-n_1)n_1}{n^2}$, while the chance that they are distinct is $\frac{n^2-n}{n^2}$. Thus
$$\frac{n^2-n}{n^2}\cdot d=\frac{n_1(n-n_1)+(n-n_1)n_1}{n^2}\quad\Longrightarrow\quad d=\frac{2p\left(1-p\right)}{1-\frac{1}{n}}.$$
\end{proof}

\begin{proof}[Proof of Theorem~\ref{theorem_general}]
Consider the anonymized rule $g^\sym$ and check that for $m=n!$ we have
\begin{equation}\label{eq_AR_MAJ_g_sym}
\AR_{n,m}\big(g^\sym\big)\geq \AR_{n,m}\big(\kMAJ{k}\big).
\end{equation}
By Lemma~\ref{lm_anonymization}, this  implies the same lower bound for $g$.

Pick $n_1\leq \frac{n}{2}$ such that the preference profile with one issue and $n_1$ supporters of the alternative $1$ is the worst profile for $\kMAJ{k}$. 

With this $n_1$, construct the profile $X$ from Lemma~\ref{lm_equidistant} and denote by $X'$ the profile, where all the voters flipped their preferences on all the issues, i.e., $x'_{i,j}=1-x_{i,j}$. Both $X$ and $X'$ are worst profiles for $\kMAJ{k}$ (see Lemma~\ref{lm_one_issue_gives_an_upper_bound}) and have the same socially-optimal cost $\SC(z^\opt(X))=\SC(z^\opt(X'))$. Therefore,
\begin{equation}\label{eq_AR_X_X_bar}
\AR_{n,m}\big(\kMAJ{k}\big)=\frac{1}{2}\frac{\SC\big(\kMAJ{k}(X)\big)+\SC\big(\kMAJ{k}(X')\big)}{\SC(z^\opt(X))}.  
\end{equation}
Our goal is to show that
\begin{equation}\label{eq_pairs_of_social_costs}
\SC\big(\kMAJ{k}(X)\big)+\SC\big(\kMAJ{k}(X')\big)\leq \SC\big(g^\sym(X)\big)+\SC\big(g^\sym(X')\big).
\end{equation}
Together with~\eqref{eq_AR_X_X_bar} this inequality would imply~\eqref{eq_AR_MAJ_g_sym} since
 $\AR_{n,m}\big(g^\sym\big)\geq \frac{\SC\big(g^\sym(Y)\big)}{\SC\big(z^\opt(Y)\big)}$ for any profile $Y$.

 Consider the distributions $g^\sym(X)$ and $g^\sym(X')$ on pairs $(C,h)$ of a committee with $k$ members and the inter-committee rule. The two distributions coincide because $g^\sym$ is distance-based and  matrices of distances for $X$ and $\bar{X}$ are the same; for definiteness, let's focus on $X$. 
 
 Since all pairs of voters in $X$ are equidistant, the only anonymous way to select a committee $C$ of size $k$ is to take it \emph{uniformly} at random among subsets $C\subset [n]$ with $|C|=k$. Indeed if two such subsets $C$ and $C'$ differ in probability, this contradicts anonymity because $C= \pi(C')$ for some permutation $\pi$ of $[n]$ and no permutation changes the matrix of distances.
 
 Consider the probability $\P(z_j=1\mid C)$ that the outcome of $g^\sym(X)$ for issue $j$ is $1$ conditional on the selected committee $C$. This probability can only depend on $(x_{i,j})_{i\in C}$ because $h$ is issue-wise.
 Since $g^\sym$ is distance-based and anonymous, and any permutation of voters (in particular, those in $C$)
 leaves the matrix of distances unchanged, we see that 
 the probability can only depend on $q=\sum_{i\in C} x_{i,j}$, i.e., the total number of supporters of the alternative $1$ for issue $j$. Hence, $\P(z_j=1\mid C)=h_{q,j}=h_{q,j}(C)$.
 
 For any alternative $j$, the chance that the uniformly random committee $C$ contains exactly $q$ supporters of the alternative $1$ is equal to ${\tiny\frac{\begin{pmatrix} n_1\\ q\end{pmatrix}\cdot\begin{pmatrix} n-n_1\\ k-q\end{pmatrix}}{\begin{pmatrix}n\\ k\end{pmatrix}}}$ for the profile $X$ and ${\tiny\frac{\begin{pmatrix} n-n_1\\ q\end{pmatrix}\cdot \begin{pmatrix} n\\ k-q\end{pmatrix}}{\begin{pmatrix}n\\ k\end{pmatrix}}}$ for $X'$. Selection of the alternative supported by majority contributes $n_1$ to the social cost for both profiles, while the sub-optimal costs $n-n_1$. This allows to rewrite the right-hand side of~\eqref{eq_pairs_of_social_costs}:
 $$\SC\big(g^\sym(X)\big)+\SC\big(g^\sym(X')\big)=$$
 $$=\sum_{j\in [m]}\sum_{q=0}^k\left({\tiny\frac{\begin{pmatrix} n_1\\ q\end{pmatrix}\cdot\begin{pmatrix} n-n_1\\ k-q\end{pmatrix}}{\begin{pmatrix}n\\ k\end{pmatrix}}}\big((n-n_1)h_{q,j}+n_1(1-h_{q,j})\big)+\right.$$
 $$+\left.{\tiny\frac{\begin{pmatrix} n-n_1\\ q\end{pmatrix}\cdot\begin{pmatrix} n_1\\ k-q\end{pmatrix}}{\begin{pmatrix}n\\ k\end{pmatrix}}}\big(n_1 h_{q,j}+(n-n_1)(1-h_{q,j})\big)\right).$$
 To minimize this expression over $h_{q,j}\in[0,1]$, one should minimize the contribution of the first summand whenever ${\tiny\frac{\begin{pmatrix} n_1\\ q\end{pmatrix}\cdot\begin{pmatrix} n-n_1\\ k-q\end{pmatrix}}{\begin{pmatrix}n\\ k\end{pmatrix}}}>{\tiny\frac{\begin{pmatrix} n-n_1\\ q\end{pmatrix}\cdot \begin{pmatrix} n\\ k-q\end{pmatrix}}{\begin{pmatrix}n\\ k\end{pmatrix}}}$ which leads to $h_{q,j}=0$; similarly, for the opposite strict inequality, it is optimal to minimize the second summand, which gives $h_{q,j}=1$. Therefore, the optimal $h_{q,j}$ equals $0$ if $q<\frac{k}{2}$ and $1$ if $q>\frac{k}{2}$. Such $h$ corresponds to the simple majority rule with arbitrary tie-breaking in case of a tie. Therefore,
  $$\SC\big(g^\sym(X)\big)+\SC\big(g^\sym(X')\big)\geq \SC\big(\kMAJ{k}(X)\big)+\SC\big(\kMAJ{k}(X')\big),$$
  which completes the proof.
\end{proof}

\end{document}